\documentclass{article}
\usepackage{authblk}
\usepackage[utf8]{inputenc}
\usepackage{fullpage}
\usepackage{amssymb}
\usepackage{amsmath}
\usepackage{amsthm}
\usepackage{physics}
\usepackage{xcolor}
\usepackage{enumitem}
\usepackage[most]{tcolorbox}
\tcbuselibrary{theorems}
\usepackage[bookmarks, colorlinks=true, urlcolor=blue, linkcolor=black, citecolor=blue]{hyperref}
\usepackage{bm}
\usepackage{tikz}
\usepackage{graphicx}

\title{Quantum policy gradient algorithms}
\author[1]{Sofiene Jerbi}
\author[2]{Arjan Cornelissen}
\author[2]{M\={a}ris Ozols}
\author[3]{Vedran Dunjko}
\affil[1]{\normalsize Institute for Theoretical Physics, University of Innsbruck}
\affil[2]{\normalsize QuSoft and University of Amsterdam}
\affil[3]{\normalsize applied Quantum algorithms (aQa), Leiden University}
\date{\today}

\newtheorem{theorem}{Theorem}[section]
\newtheorem{lemma}[theorem]{Lemma}
\newtheorem{definition}[theorem]{Definition}
\newtheorem{corollary}[theorem]{Corollary}

\newtcbtheorem[number within=section]{thm}{Theorem}%
{colframe=gray!50!black,fonttitle=\bfseries}{}

\makeatletter
\newtheorem{repeatlem@}{Lemma}
\newenvironment{repeatlem}[1]{%
    \def\therepeatlem@{\ref{#1}}
    \repeatlem@
}
{\endrepeatlem@}
\makeatother

\newcommand{\A}{\ensuremath{\mathcal{A}}}
\newcommand{\B}{\ensuremath{\mathcal{B}}}
\newcommand{\E}{\ensuremath{\mathbb{E}}}
\newcommand{\G}{\ensuremath{\mathcal{G}}}
\renewcommand{\H}{\ensuremath{\mathcal{H}}}
\newcommand{\M}{\ensuremath{\mathcal{M}}}
\newcommand{\N}{\ensuremath{\mathbb{N}}}
\renewcommand{\O}{\ensuremath{\mathcal{O}}}
\renewcommand{\P}{\ensuremath{\mathcal{P}}}
\newcommand{\Pbb}{\ensuremath{\mathbb{P}}}
\newcommand{\R}{\ensuremath{\mathbb{R}}}
\renewcommand{\S}{\ensuremath{\mathcal{S}}}
\newcommand{\X}{\ensuremath{\mathcal{X}}}
\newcommand{\Z}{\ensuremath{\mathbb{Z}}}
\newcommand{\alphas}{\ensuremath{\bm{\alpha}}}
\newcommand{\omegas}{\ensuremath{\bm{\omega}}}
\newcommand{\phis}{\ensuremath{\bm{\phi}}}

\newcommand{\params}{\ensuremath{\bm{\theta}}}
\newcommand{\weights}{\ensuremath{\bm{w}}}
\newcommand{\policy}{\ensuremath{\pi_{\params}}}
\newcommand{\valuefct}{\ensuremath{V_{\policy}}}
\newcommand{\Rmax}{\ensuremath{\abs{R}_\textrm{max}}}
\newcommand{\Rwd}{\ensuremath{\mathcal{R}}}
\renewcommand{\grad}{\ensuremath{\nabla_{\params}}}

\usepackage[bottom]{footmisc}

\begin{document}

\maketitle

\begin{abstract}
\noindent Understanding the power and limitations of quantum access to data in machine learning tasks is primordial to assess the potential of quantum computing in artificial intelligence. Previous works have already shown that speed-ups in learning are possible when given quantum access to reinforcement learning environments. Yet, the applicability of quantum algorithms in this setting remains very limited, notably in environments with large state and action spaces. In this work, we design quantum algorithms to train state-of-the-art reinforcement learning policies by exploiting quantum interactions with an environment. However, these algorithms only offer full quadratic speed-ups in sample complexity over their classical analogs when the trained policies satisfy some regularity conditions. Interestingly, we find that reinforcement learning policies derived from parametrized quantum circuits are well-behaved with respect to these conditions, which showcases the benefit of a fully-quantum reinforcement learning framework.
\end{abstract}

\section{Introduction}

When studying the potential advantages of quantum computing in machine learning, a natural question that arises is whether quantum algorithms that exploit \emph{quantum access} to data can speed up learning.\break In the context of supervised learning, this led to the development of algorithms based on quantum RAMs, which can achieve high-degree polynomial improvements over their classical analogs \cite{chia22}. In reinforcement learning, where we consider learning agents interacting with task environments, the question becomes: can quantum interactions with an environment, and in particular the ability to explore several trajectories in superposition, be beneficial for a learning agent. In recent years, several works have approached this question from a variety of angles \cite{dunjko17}: based on Grover's algorithm \cite{grover96}, some works have for instance shown that searching for an optimal sequence of actions in an environment can be done using quadratically fewer interactions given the appropriate oracular access to the environment \cite{dunjko16,saggio21,hamann21}. Other works have considered the more general problem of finding the optimal policy in a Markov Decision Process (MDP), and have found that up to quadratic speed-ups in the number of interactions are also possible, again given the proper oracular access \cite{wang21,wang21b,ronagh19,cherrat22}. Finally, tailored MDP environments (based, e.g., on Simon's problem) have also been introduced, which allow for exponential quantum speed-ups in learning times compared to the best classical agents \cite{dunjko17b}.

Yet, all the quantum algorithms that have been proposed in this quantum-accessible setting remain inefficient in the most well-publicised use cases of reinforcement learning, such as Go \cite{silver17}, city navigation \cite{mirowski18}, and computer games \cite{mnih15}: environments with large state-action spaces. Indeed, aside from the task-specific algorithms of Ref.~\cite{dunjko17b}, the proposed algorithms scale at best as the square root of the size of the state-action space, which is intractable in most modern-day applications that deal for instance with image-based inputs. In the classical literature, modern approaches to reinforcement learning in large spaces commonly replace the explicit storage of a policy (and/or a value function) in a table of values by a parametrized model (e.g., a deep neural network), whose parameters $\params$ have a much smaller size than the state-action space. One of the earliest approaches based on such parametrized models is that of \emph{policy gradient algorithms} \cite{williams92,sutton00}. This approach frames reinforcement learning as a direct optimization problem, where the expected rewards (or value function) $\valuefct(s_0)$ of a given policy $\policy$ starting its interactions in a state $s_0$ is optimized via gradient ascent on the policy parameters $\params$. Therefore, the core task in this approach is to estimate the gradient $\grad\valuefct(s_0)$ to a certain error $\varepsilon$ in the $\ell_\infty$-norm. For this task, two approaches are common: \emph{numerical} gradient estimation \cite{kohl04}, where the value function is evaluated at different parameter settings $\params'$ centered around $\params$, that are combined to estimate the gradient at $\params$ (using, e.g., a central difference method), and \emph{analytical} gradient estimation \cite{sutton00}, using a formulation of this gradient as a function of the rewards and the gradients of the policy $\policy$, averaged over trajectories generated by $\policy$ (i.e., a Monte Carlo method).

Concurrently in the last few years, several works have introduced quantum parametrized models, known most commonly as parametrized or variational quantum circuits, that could take the place of deep neural networks in both policy-based \cite{jerbi21,sequeira22,chen22} and value-based \cite{chen20,lockwood20,wu20,skolik21} reinforcement learning. While evaluated on a quantum computer, these models are however trained via classical interaction with the environment using, e.g., a classical policy gradient method.

In this work, we present quantum algorithms that speed up both the numerical and analytical gradient estimation approaches to policy gradient methods. These algorithms exploit an appropriately defined oracular access to the environment that allows to explore several trajectories in superposition, combined with subroutines for numerical gradient estimation \cite{gilyen19,cornelissen19} and multivariate Monte Carlo estimation \cite{cornelissen21,cornelissen22}. Both these subroutines are however known to offer full quadratic speed-ups only in certain regimes, that depend in our setting on the smoothness of the value function $\valuefct(s_0)$ and on the $\ell_p$-norm of its gradient $\grad\valuefct(s_0)$, respectively. Conveniently, we also identify families of parametrized quantum policies $\policy$ previously studied in the literature \cite{jerbi21} that satisfy the conditions of these regimes. We therefore end up with quantum policy gradient algorithms to train quantum policies, i.e., a fully quantum approach to reinforcement learning in large spaces.

\section{Preliminaries}

In this section, we present the main tools and concepts that we need to design our quantum policy gradient algorithms. We start by introducing policy gradient methods in Sec.~\ref{sec:PGM}. We then define the general oracle types that we consider in this work in Sec.~\ref{sec:input-models}, which allows us to properly define the notion of quantum access to a reinforcement learning environment in Sec.~\ref{sec:quantum-access-env}. We define the parametrized quantum policies that we apply our quantum policy gradient algorithms to in Sec.~\ref{sec:quantum-policies}. And finally, we present the core subroutines used in our quantum algorithms in Sec.~\ref{sec:core-subroutines}.

\subsection{Policy gradient methods\label{sec:PGM}}

At the core of policy gradient methods are two ingredients: a parametrized policy $\policy$, that governs an agent's actions in an environment, and its associated value function $\valuefct$, which evaluates the long-term performance of this policy in the environment. The policy $\policy(\cdot|s)$ is a conditional probability distribution over actions given a state $s$, parametrized by a vector of parameters $\params \in \mathbb{R}^d$. When acting with a given policy in the environment, the agent experiences sampled trajectories (or episodes) $\tau = (s_0, a_0, r_0, s_1, \ldots)$ composed of states, actions and rewards that depend both on the policy of the agent and the environment dynamics (see Sec.~\ref{sec:quantum-access-env} for more details). The standard figure of merit used to assess the performance of a policy $\policy$ is called the value function $\valuefct(s_0)$ and is given by the expected sum of rewards (or return) $R(\tau)$ collected in a trajectory:
\begin{equation}\label{eq:value-function}
    \valuefct(s_0) = \mathbb{E}_{\policy, P_{E}} \left[ \sum_{t=0}^{T-1} \gamma^t r_t \right] = \mathbb{E}_{\policy, P_{E}} \left[ R(\tau) \right]
\end{equation}
where $s_0$ is the initial state of the agent's interaction $\tau$ with the environment and $P_E$ a description of the environment dynamics (e.g., in the form of an MDP, see Def.~\ref{def:MDP}). Each episode of interaction has a horizon (or length) $T \in \mathbb{N} \cup \{\infty\}$ and the returns $R(\tau)$ involve a discount factor $\gamma \in [0,1]$ that allows, when $\gamma < 1$, to avoid diverging value functions for an infinite horizon, i.e., $T = \infty$.

Policy gradient methods take a direct optimization approach to RL: starting from an initial policy $\policy$, its parameters are iteratively updated such as to maximize its associated value function $\valuefct(s_0)$, via gradient ascent. For this method to be applicable, one needs to evaluate the gradient of the value function $\grad\valuefct$, up to some error $\varepsilon$ in $\ell_\infty$-norm to be specified.

\subsubsection{Numerical gradient estimation}

The most straightforward approach to estimate the value function of a policy is via a Monte Carlo approach: by collecting $N$ episodes $\tau_i = (s_0^{(i)}, a_0^{(i)}, r_0^{(i)}, s_1^{(i)} \ldots)$ governed by $\policy$, one can compute for each of these the discounted return $R(\tau)$ appearing in Eq.\ (\ref{eq:value-function}) and average the results. The resulting value
\begin{equation}\label{eq:mc-value-function}
	\widetilde{V}_{\policy}(s) =\frac{1}{N}\sum_{i=1}^N \sum_{t=0}^{T-1} \gamma^t r_{t}^{(i)}
\end{equation}
is a Monte Carlo estimate of the value function.

With the capacity to estimate the value function, we can also estimate its gradient using numerical methods. In its simplest form, a finite-difference method simply evaluates $\widetilde{V}_{\policy}(s_0)$ and $\widetilde{V}_{\pi_{\bm{\theta + \delta e_i}}}(s_0)$ for $\delta>0$ and $e_i = (0,\ldots,0,1_i,0\ldots,0)$ a unit vector with support on the $i$-th parameter in $\params$, and returns the estimate:
\begin{equation}
	\partial_i \valuefct(s_0) \approx \frac{\widetilde{V}_{\pi_{\bm{\theta + \delta e_i}}}(s_0)-\widetilde{V}_{\policy}(s_0)}{\delta}.
\end{equation}
Even though more elaborate finite difference methods exist (that we will use in Sec.~\ref{sec:num-grad}), they inherently have a sample complexity (in terms of the number of interactions with the environment) that scales linearly in the dimension of $\params$.

\subsubsection{Analytical gradient estimation\label{sec:analytical-grad-prelim}}

Perhaps one of the most appealing aspects of policy gradient methods is that the gradients of value functions also have an analytical formulation whose evaluation has a sample complexity only logarithmic in the dimension of $\params$ \cite{kakade03}. This analytical formulation is known as the policy gradient theorem:

\begin{theorem}[Policy gradient theorem \cite{sutton00}]
	Given a policy $\policy$ that generates trajectories $\tau = (s_0, a_0, r_0, s_1, \ldots)$ in a reinforcement learning environment with time horizon $T\in\N\cup\{\infty\}$, the gradient of the value function $\valuefct$ with respect to $\params$ is given by
	\begin{equation}
		\grad\valuefct (s_0) = \mathbb{E}_{\tau}  \left[ \sum_{t=0}^{T-1}\grad\log{\policy(a_t|s_t)} \sum_{t'=0}^{T-1} \gamma^{t'} r_{t'} \right].
	\end{equation}
\end{theorem}

\noindent A simple derivation of this Theorem can be found in Appendix \ref{appdx:derivation-pgt}. Essentially, due to the so-called ``log-likelihood trick'' \cite{silver15}, the differentiation with respect to the policy parameters can be made to act solely on the random variables ``inside'' the expected value, while leaving the probability distribution behind this expected value unchanged. This means that the gradient of the value function can, similarly to the value function itself, be estimated via Monte Carlo sampling of trajectories governed by a fixed $\policy$ and environment-independent computations (i.e., the evaluation of $\grad\log{\policy(a_t|s_t)}$).

\subsection{Input models\label{sec:input-models}}

To design our quantum algorithms, we need to define access models to the environment as well as the policy $\policy$ to be trained. We do this in terms of oracles that can be queried in superposition. Throughout this manuscript, we will be dealing with several types of such oracles, all defined in this section. 

\begin{definition}[Oracle types]\label{def:oracles}
Let $\X$ be a finite set whose elements $x\in\X$ can be encoded as mutually orthogonal states $\ket{x}$, and let $f:\X \mapsto [0,B]$ be a function acting on this set, whose output is bounded by some $B\in\R$. We define different types of oracle access to $f$:
\begin{enumerate}
	\item \textbf{Binary oracle}: $f(x)$ is encoded in an additional register using a binary representation of a desired precision:
	\begin{equation}
		\B_f : \ket{x}\ket{0} \mapsto \ket{x}\ket{f(x)},
	\end{equation}
	\item \textbf{Phase oracle}: $f(x)$ is encoded in the phase of the input register:
	\begin{equation}
		O_f : \ket{x} \mapsto e^{i\frac{f(x)}{B}}\ket{x},
	\end{equation}
	\item \textbf{Probability oracle}: $f(x)$ is encoded in the amplitude of an additional qubit (possibly entangled to arbitrary states $\ket{\psi_{0}(x)}$ and $\ket{\psi_{1}(x)}$ of an additional register):
	\begin{equation}
		\widetilde{O}_f : \ket{x}\ket{0}\ket{0} \mapsto \ket{x}\left(\sqrt{\frac{f(x)}{B}}\ket{0}\ket{\psi_0(x)} + \sqrt{1-\frac{f(x)}{B}}\ket{1}\ket{\psi_1(x)}\right).
	\end{equation}
\end{enumerate}
\end{definition}

Clearly, having access to a binary oracle $\B_f$, we can easily convert it into a phase or probability oracle $O_f$ or $\widetilde{O}_f $, using one call to $\B_f$ first, then a single-qubit rotation or a phase gate controlled on $\ket{f(x)}$, and finally a call to $\B_f^{\dagger}$ to uncompute $\ket{f(x)}$.

We will also need a subroutine to convert probability oracles into phase oracles:
\begin{lemma}[Probability to phase oracle (Corollary 4.1 in \cite{gilyen19})]\label{lem:prob-to-phase}
Suppose that we are given a probability oracle $\widetilde{O}_f$ for $f:\X \rightarrow [0,B]$. We can implement a phase oracle $O_f$ up to operator norm error $\varepsilon$, with query complexity $\O(\log(1/\varepsilon))$, i.e., this many calls to $\widetilde{O}_f$ and its inverse.
\end{lemma}

\subsection{Quantum-accessible environments\label{sec:quantum-access-env}}

Inspired by previous work that considered the quantum-accessible reinforcement learning setting \cite{dunjko17b,wang21,wang21b,ronagh19,cherrat22}, we define oracular access to a specific type of reinforcement learning environments called Markov Decision Processes (MDPs) \cite{sutton98}, defined as follows:

\begin{definition}[Markov Decision Process (MDP)]\label{def:MDP}
A Markov Decision Process is defined by a tuple $(\S,\A,P,R,\Rmax,T,\gamma)$, where $\S$ is a finite state space, $\A$ is a finite action space, $P: \S\times\A\times\S \rightarrow [0,1]$ is a transition probability matrix with entries $P(s'|s,a)$ that govern the transition to a state $s'\in\S$ after performing action $a\in\A$ in state $s\in\S$, $R:\S\times\A \rightarrow [-\Rmax,\Rmax]$ is a reward function bounded by $\Rmax\in\R_+$ that assigns a reward $R(s,a)$ to every state-action pair, $T\in\N\cup\{\infty\}$ is a (possibly infinite) time horizon for each episode of interaction, and $\gamma\in[0,1]$ is a discount factor, with the restriction that $\gamma<1$ for $T=\infty$.
\end{definition}

\noindent Our oracular access to the environment takes the form of two oracles that coherently implement the MDP dynamics:

\begin{definition}[Quantum access to an MDP]\label{def:quantum-access-MDP}
Let $\M = (\S,\A,P,R,\Rmax,T,\gamma)$ be an MDP as defined in Def.~\ref{def:MDP}. We say that we have quantum access to the MDP if we can call the following oracles:
\begin{enumerate}
	\item An oracle $\P$ that coherently samples a column of the transition probability matrix $P$:
		\begin{equation}
			\P : \ket{s,a}\ket{0} \mapsto \ket{s,a} \sum_{s' \in \S} \sqrt{P(s'|s,a)}\ket{s'}.
		\end{equation}
	\item An oracle $\Rwd$ that returns a binary representation of the output of the reward function $R$:
		\begin{equation}
			\Rwd : \ket{s,a}\ket{0} \mapsto \ket{s,a}\ket{R(s,a)}.
		\end{equation}
\end{enumerate}
\end{definition}

\noindent We also assume the ability to construct a unitary $\Pi$ that coherently implements a policy $\policy$:

\begin{definition}[Quantum evaluation of a policy]\label{def:quantum-eval-policy}
Let $\policy: \S\times\A \rightarrow [0,1]$ be a reinforcement learning policy acting in a state-action space $\S\times\A$ and parametrized by a vector $\params \in \R^{d}$ (that can be encoded with finite precision as $\ket{\params}$). We say that the policy is quantum-evaluatable if we can construct a unitary satisfying:
\begin{equation}
	\Pi : \ket{\params}\ket{s}\ket{0} \mapsto \ket{\params}\ket{s} \sum_{a \in \A} \sqrt{\policy(a|s)}\ket{a}.
\end{equation}
\end{definition}

\noindent Such a construction would be very natural for some quantum policies (such as the \textsc{raw-PQC} defined in the next subsection). But any policy that can be computed classically could also be turned into such a unitary via quantum simulation of the classical computation of $(\policy(a|s) : a\in\A)$ and known subroutines to encode this probability vector into the amplitudes of a quantum state \cite{grover02}.

Equipped with the proper quantum access to the environment and the policy, we can construct simple subroutines that create superpositions of trajectories in the environment and evaluate the returns of these trajectories.

\begin{lemma}[Superposition of trajectories]\label{lem:U-traj}
Let $\M$ be a quantum-accessible MDP with oracles $\P,\Rwd$ as defined in Def.~\ref{def:quantum-access-MDP}, and let $\policy$ be a quantum-evaluatable policy with its unitary implementation $\Pi$ as defined in Def.~\ref{def:quantum-eval-policy}. A unitary that prepares a coherent superposition of all trajectories $\tau = (s_0, a_0, \ldots, s_{T-1}, a_{T-1})$ of length $T$ (without their rewards), i.e.,
\begin{equation}
	U_{P(\tau)} : \ket{\params}\ket{s_0}\ket{0} \mapsto \ket{\params}\sum_{\tau} \sqrt{P_{\params}(\tau)} \ket{s_0, a_0, \ldots, s_{T-1}, a_{T-1}}
\end{equation}
for $P_{\params}(\tau) = \prod_{t=0}^{T-1}\policy(a_t|s_t)P(s_{t+1}|s_t,a_t)$, can be implemented using $\O(T)$ calls to $\P$ and $\Pi$.
\end{lemma}
\begin{proof}
We apply sequentially $\Pi$ and $\P$ on the registers indexed $\{0,2i+1,2i+2\}$ and $\{2i+1,2i+2,2i+3\}$ respectively, for $i = 0, \ldots, T-1$. This amounts to $T$ calls to each oracle.
\end{proof}

\begin{lemma}[Return]\label{lem:U-ret}
Let $\M$ be a quantum-accessible MDP with oracles $\P,\Rwd$ as defined in Def.~\ref{def:quantum-access-MDP}, and let $\tau = (s_0, a_0, \ldots, s_{T-1}, a_{T-1})$ be a trajectory of length $T$ in this MDP (without its rewards). A unitary that computes the return $R(\tau) = \sum_{t=0}^{T-1} \gamma^t r_t$ associated to this trajectory, i.e.,
\begin{equation}
	U_{R(\tau)} : \ket{\tau}\ket{0} \mapsto \ket{\tau}\ket{R(\tau)}
\end{equation}
can be implemented using $\O(T)$ calls to $\Rwd$.
\end{lemma}
\begin{proof}
Using $T$ calls to $\Rwd$, we simply collect all the rewards of the trajectory in an additional register. Then we simulate a classical circuit that computes the discounted sum of these rewards $R(\tau)$ (then uncompute the rewards using $T$ calls to $\Rwd$ on the same register).
\end{proof}

\begin{figure}[t]
	\centering
	\includegraphics[width=1\linewidth]{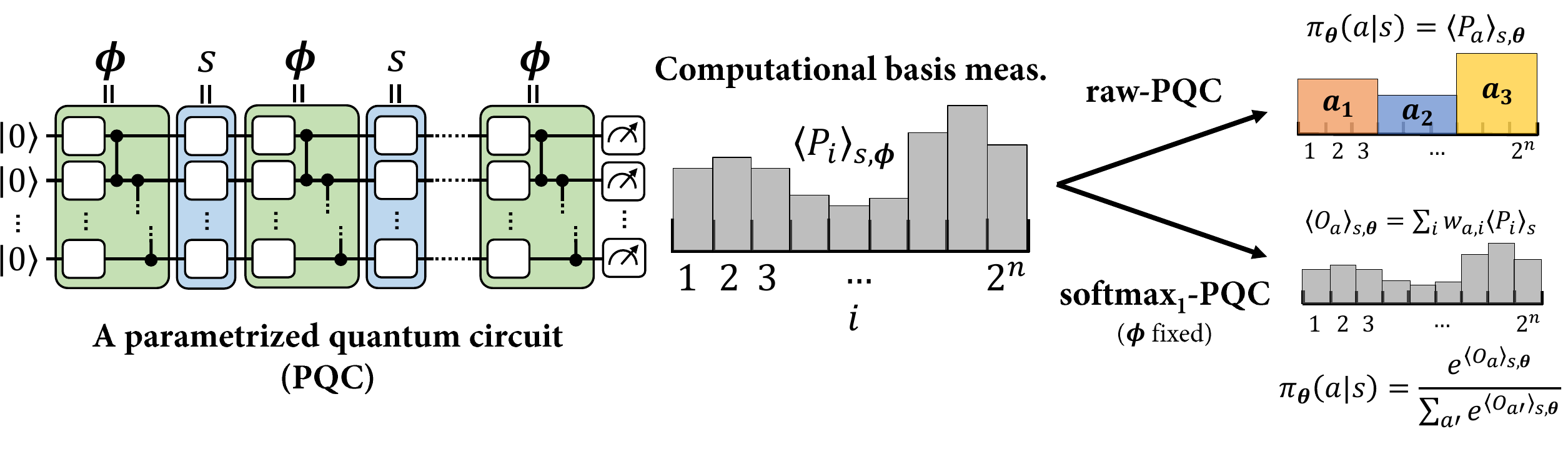}
	\caption{\textbf{The parametrized quantum policies considered in this work.} A parametrized quantum circuit (PQC) taking as input the agent's state $s$ and parameters $\phis$ produces a quantum state which has probability $\expval{P_i}_{s,\phis}$ of being projected onto the (computational) basis state $\ket{i}$. The \textsc{raw-PQC} policy simply assigns a subset of these basis states to each action $a\in\A$, and its parameters are $\params = \phis$. The $\textsc{softmax}_1$\textsc{-PQC} policy uses instead a fixed assignment of $\phis$, and computes the weighted expectation values $\expval{O_a}_{s,\params} = \sum_i w_{a,i} \expval{P_i}_{s}$.\protect\footnotemark\ The softmax of these expectation values gives the policy $\policy$, whose parameters are $\params = \bm{w}$.}
	\label{fig:PQC-policies}
\end{figure}

\subsection{Quantum policies\label{sec:quantum-policies}}

The efficiency of our quantum policy gradient algorithms depends on regularity conditions on the policies $\policy$ to be trained. Particularly well-behaved policies are policies defined out of parametrized quantum circuits (PQC) \cite{benedetti19} that have been previously studied in classical reinforcement learning environments \cite{jerbi21}. For each of our numerical and analytical gradient estimation algorithms, we will be interested more specifically in a certain type of PQC-policies, depicted in Fig.~\ref{fig:PQC-policies}, and defined below.

\begin{definition}[\textsc{raw-PQC}]\label{def:raw-PQC}
Given a PQC acting on $n$ qubits, taking as input a state $s \in \S$ and parameters $\phis \in \mathbb{R}^d$, such that its corresponding unitary $U(s,\phis)$ produces the quantum state $\ket{\psi_{s,\phis}} = U(s,\phis) \ket{0^{\otimes n}}$, we define its associated \textsc{raw-PQC} policy as:
	\begin{equation}
	\policy (a|s) = \expval{P_a}_{s,\params}
	\end{equation}
where $\expval{P_a}_{s,\params} = \ev{P_{a}}{\psi_{s,\phis}}$ is the expectation value of a projection $P_{a}$ associated to action $a$, such that $\sum_{a} P_a = I$ and $P_aP_{a'} = \delta_{a,a'}P_a$. $\params = \phis$ constitutes all of its trainable parameters.
\end{definition}

\begin{definition}[\textsc{softmax-PQC}]\label{def:softmax-PQC}
Given a PQC acting on $n$ qubits, taking as input a state $s \in \S$ and parameters $\phis \in \mathbb{R}^{d'}$, such that its corresponding unitary $U(s,\phis)$ produces the quantum state $\ket{\psi_{s,\phis}} = U(s,\phis) \ket{0^{\otimes n}}$, we define its associated \textsc{softmax-PQC} policy as:
	\begin{equation}\label{eq:softmax-PQC}
	\policy (a|s) = \frac{e^{\expval{O_a}_{s,\params}}}{\sum_{a'} e^{\expval{O_{a'}}_{s,\params}}}
	\end{equation}
where $\expval{O_a}_{s,\params} = \ev{\sum_i w_{a,i} H_{a,i}}{\psi_{s,\phis}}$ is the expectation value of the weighted Hermitian operators $H_{a,i}$ associated to action $a$ with weights $w_{a,i}\in\R$. $\params = (\phis,\weights)$ constitutes all of its trainable parameters.
\end{definition}

\footnotetext{Note that the choice of basis for the measurement, i.e., the $P_i$'s, could also depend on $a$.}

\noindent More specifically, we are interested in a restricted family of \textsc{softmax-PQC} policies:

\begin{definition}[$\textsc{softmax}_1$\textsc{-PQC}]\label{def:softmax1-PQC}
We define a $\textsc{softmax}_1$\textsc{-PQC} policy as a \textsc{softmax-PQC} where $\phis = \varnothing$ and, for all $a\in\A$, $H_{a,i} = P_{a,i}$ is a projection on a subspace indexed by $i$, such that $\sum_{i} P_{a,i} = I$ and $P_{a,i}P_{a,i'} = \delta_{i,i'}P_{a,i}$.\footnote{This constraint includes the degenerate case where $P_{a,i} = P_{a',i} = P_i$, for all $a,a'$, illustrated in Fig.~\ref{fig:PQC-policies}.}
\end{definition}

\noindent We call the resulting policy a $\textsc{softmax}_1$\textsc{-PQC}, as its log-policy gradient is always bounded in $\ell_1$-norm, i.e., $\norm{\grad\log\policy(a|s)}_1 \leq 2, \forall s,a,\params$ (see Lemma \ref{lem:softmax1-PQC-bound}).

\subsection{Core subroutines\label{sec:core-subroutines}}

The core methods behind numerical and analytical policy gradient algorithms both have their quantum analogs, that offer up to quadratic speed-ups in certain regimes. In this section, we present these quantum subroutines and explain the conditions that govern the speed-up regimes.

\subsubsection{Quantum gradient estimation\label{sec:num-grad-sub}}

Quantum algorithms for gradient estimation have been studied since early works in quantum computing. Notably, Jordan's algorithm \cite{jordan05} manages to estimate gradients $\nabla_{\params} f(\params)$ with a query complexity that is independent of their dimension $d = \abs{\params}$. However, this algorithm assumes a very powerful binary oracle access to the input function $f$ (see Def.~\ref{def:oracles}). And for functions that cannot be evaluated to arbitrary precision $\varepsilon$ with a negligible cost in $\varepsilon^{-1}$ (e.g., $\O(1)$ or $\O(\log(\varepsilon^{-1}))$), which is the case of value functions, the construction of this oracle introduces non-negligible costs \cite{gilyen19}. More precisely, these costs depend on the dimension $d$, but also on the smoothness of the derivatives of $f$, as smoother functions are more amenable to efficient evaluation of their gradient. Notably, a measure of smoothness that has been studied for quantum gradient estimation is the Gevrey condition \cite{gilyen19,cornelissen19}:

\begin{definition}[Gevrey functions]\label{def:gevrey_functions}
Let $d \in \N$, $\sigma \in [0,1]$, $M > 0$, $c > 0$, $\Omega \subseteq \R^d$ an open subset and $f : \R^d \to \R$. We say that $f$ is a Gevrey function on $\Omega$ with parameters $M$, $c$ and $\sigma$, and denote $f\in\G_{d,M,c,\sigma,\Omega}$ when all (higher order) partial derivatives of $f$ exist, and the following upper bound on its partial derivatives is satisfied for all $\bm{x} \in \Omega$, $k \in \N_0$ and $\alphas \in [d]^k$:
\begin{equation}
	\abs{\partial_{\alphas}f(\bm{x})} \leq \frac{M}{2}c^k(k!)^{\sigma}.
\end{equation}
\end{definition}

\noindent The query complexity of the quantum gradient estimation algorithm is summarized in the following theorem:

\begin{theorem}[Numerical gradient estimation (Theorem 3.8 in \cite{cornelissen19})]\label{thm:num-gradient}
Given phase oracle access $O_f$ to a function $f \in \G_{d,M,c,\sigma,\Omega}$, an $\varepsilon \in (0,c)$, and an $\bm{x}\in\Omega$ (such that a hypercube of edge length $\O(\log(cd^{\sigma}/\varepsilon)/\varepsilon)$ centered around $\bm{x}$ is still in $\Omega$), there exists an algorithm that returns an $\varepsilon$-precise estimate of $\nabla f(\bm{x})$ in $\ell_\infty$-norm with success probability at least $2/3$ using
\begin{equation}
	\widetilde{\O}\left(\frac{Mcd^{\max\{\sigma,1/2\}}}{\varepsilon}\right)
\end{equation}
queries to $O_f$.
\end{theorem}

\noindent Notably, in this case the dependence on the dimension of the gradient can only be reduced to $\sqrt{d}$ when the Gevrey condition of $f$ satisfies $\sigma \leq 1/2$.

\subsubsection{Quantum multivariate Monte Carlo\label{sec:quantum-MC-sub}}

Quantum algorithms for estimating the mean $\E[X]$ of a \emph{univariate} random variable $X$ taking values in $\R$ \cite{montanaro15}\break have been studied since early works by Grover \cite{grover98}, and culminated to a near-optimal algorithm that outperforms any classical estimator \cite{hamoudi21}. However, the case of \emph{multivariate} random variables $X$ taking values in $\R^d$ has been studied only more recently \cite{cornelissen21,cornelissen22,huggins21}, and exhibits a dependence on the dimension $d$ that can be up to exponentially worse than for classical estimators (which is $\O(\log(d))$, see Lemma \ref{lem:classical-MVMC}). Before presenting explicitly this dependence on $d$, we first define the input model we consider for this problem:

\begin{definition}[Quantum samples]\label{def:quantum-exp}
Consider a finite random variable $X:\Omega \rightarrow E$ on a probability space $(\Omega, 2^{\Omega}, P)$. Let $\H_{\Omega}$ and $\H_E$ be two Hilbert spaces with basis states $\{\ket{\omega}\}_{\omega \in \Omega}$ and $\{\ket{x}\}_{x \in E}$ respectively.\break We say that we have quantum-sample access to $X$ when we can call the two following oracles:
\begin{enumerate}
	\item A unitary $U_{P}$ acting on $\H_{\Omega}$ as:
	\begin{equation}
		U_{P} : \ket{0} \mapsto \sum_{\omega \in \Omega} \sqrt{P(\omega)} \ket{\omega}
	\end{equation}
	and its inverse~$U_{P}^{-1}$.
	\item A binary oracle $\B_X$ acting on $\H_{\Omega} \otimes \H_E$ such that:
	\begin{equation}
		\B_X : \ket{\omega}\ket{0} \mapsto \ket{\omega}\ket{X(\omega)}.
	\end{equation}
\end{enumerate}
\end{definition}

\newpage

\begin{theorem}[Multivariate Monte Carlo estimation (Theorem 3.3 in \cite{cornelissen22})]\label{thm:MVMC}
Let $X$ be a $d$-dimensional bounded random variable such that $\norm{X}_p \leq B$ for some $p\geq1$. Given quantum-sample access to $X$, for any $\varepsilon,\delta >0$, there exists a quantum multivariate mean estimator that returns an $\varepsilon$-precise estimate of $\E[X]$ in $\ell_\infty$-norm with success probability at least $1-\delta$ using
\begin{equation}
	\widetilde{\O}\left(\frac{B d^{\xi(p)}}{\varepsilon}\right)
\end{equation}
queries to $X$, where $\xi(p) = \max\{0,\frac{1}{2}-\frac{1}{p}\}$.
\end{theorem}
In contrast to the exposition of Theorem 3.3 in \cite{cornelissen22}, we have used Hölder's inequality $\norm{X}_2 \leq d^{\xi(p)}\norm{X}_p$ to make use of a bound on $X$ in any $\ell_p$-norm, renormalized $X$ by $d^{\xi(p)}B$ (a factor which reappears linearly in the number of oracle calls needed, as it impacts linearly the precision needed), and trivially upper bounded $\E[\norm{X}_2]$ by $L_2 = 1$.

\section{Numerical gradient estimation\label{sec:num-grad}}

We obtain our numerical policy gradient algorithm from the quantum gradient estimation subroutine introduced in Sec.~\ref{sec:num-grad-sub}. For this, we need to construct a phase oracle to the value function $\valuefct(s_0)$, which can easily be obtained from the unitaries $U_{P(\tau)}$ and $U_{R(\tau)}$ constructed in Lemma \ref{lem:U-traj} and \ref{lem:U-ret} (see below). But we also need to show that the value function satisfies a Gevrey condition $\sigma \leq 1/2$ in order to get a full quadratic speed-up in sample complexity. For this, we identify the quantity:
\begin{equation}\label{eq:D}
	D = \max_{k\in\N^*} (D_k)^{1/k}
\end{equation}
where $\N^* = \N\backslash\{0\}\cup\{\infty\}$ and
\begin{equation}
	D_k = \max_{s\in\S, \alphas\in[d]^k} \sum_{a\in\A}\abs{\partial_{\alphas} \policy(a|s)}.
\end{equation}
which we show governs the Gevrey condition of the value function. More precisely, we find in Lemma \ref{lem:gevrey-valuefct} that it satisfies $\sigma = 0, M=\frac{4\Rmax}{1-\gamma}$ and $c=DT^2$ in Def.~\ref{def:gevrey_functions}. This allows us to show the following Theorem:\break

\begin{thm}{Numerical policy gradient algorithm}{}
Let $\policy$ be a policy parametrized by a vector $\params\in\R^d$, that can be used to interact with a quantum-accessible MDP $\M=(\S,\A,P,R,\Rmax,T,\gamma)$ with $\gamma T \geq 2$, and such that $\policy$ has a bounded smoothness parameter $D$, defined in Eq.~(\ref{eq:D}). The gradient of the value function corresponding to this policy, $\grad\valuefct(s_0)$, can be evaluated to error $\varepsilon$ in $\ell_\infty$-norm, using
\begin{equation}
	\widetilde{\O}\left(\sqrt{d} \frac{DT^2\Rmax}{\varepsilon(1-\gamma)}\right)
\end{equation}
length-$T$ episodes of interaction with the environment using a quantum numerical gradient estimator, while a classical numerical gradient estimator needs
\begin{equation}
	\widetilde{\O}\left( d \left(\frac{DT^2\Rmax}{\varepsilon(1-\gamma)}\right)^2\right)
\end{equation}
length-$T$ episodes of interaction with the environment.
\end{thm}
\begin{proof}
We apply Theorem \ref{thm:num-gradient} for $f=\valuefct(s_0)$ as a function of $\params$. To construct the phase oracle $O_f$, we first construct a probability oracle $\widetilde{O}_f$ to $f$. For this we apply on the state $\ket{s_0}\ket{0}$ the unitaries $U_{P(\tau)}$ and $U_{R(\tau)}$ from Lemmas \ref{lem:U-traj} and \ref{lem:U-ret} respectively, to get
\begin{equation}
	\ket{\params}\ket{s_0}\ket{0}\ket{0} \mapsto \ket{\params}\sum_{\tau} \sqrt{P_{\params}(\tau)} \ket{\tau} \ket{R(\tau)} \ket{0}.
\end{equation}
Then we rotate the last qubit proportionally to the return $R(\tau)$, such that the probability of this qubit being $\ket{0}$ encodes the value function:
\begin{align}
	&\mapsto \ket{\params}\sum_{\tau} \sqrt{P_{\params}(\tau)} \ket{\tau} \ket{R(\tau)} \left(\sqrt{\widetilde{R}(\tau)}\ket{0} + \sqrt{1-\widetilde{R}(\tau)}\ket{1}\right)\\
	&= \ket{\params}\sqrt{\widetilde{V}_{\policy}(s_0)} \ket{\psi_0}\ket{0} + \sqrt{1-\widetilde{V}_{\policy}(s_0)} \ket{\psi_1}\ket{1}
\end{align}
where $\widetilde{R}(\tau) = \frac{R(\tau)(1-\gamma)}{\Rmax}$ and $\widetilde{V}_{\policy}(s_0) = \frac{\valuefct(s_0)(1-\gamma)}{\Rmax}$. This probability oracle $\widetilde{O}_f$ can be converted into a phase oracle $O_f$ using Lemma \ref{lem:prob-to-phase}, which only comes with a logarithmic overhead in the query complexity.\\
From Lemma \ref{lem:gevrey-valuefct}, we know that the value function satisfies the Gevrey conditions for $\sigma = 0, M=\frac{4\Rmax}{1-\gamma}$ and $c=DT^2$,  in Theorem \ref{thm:num-gradient}, resulting in the stated quantum query complexity.\\

The classical query complexity is proven in Lemma \ref{lem:classical-num}.
\end{proof}

Note that the total query complexity of the quantum and classical numerical gradient estimators, in terms of the number of calls to $\P$ and $\Rwd$, is $\widetilde{\O}\left( \sqrt{d} \frac{DT^3\Rmax}{\varepsilon(1-\gamma)}\right)$ and $\widetilde{\O}\left( d \frac{D^2T^5\Rmax^2}{\varepsilon^2(1-\gamma)^2}\right)$, respectively.

\noindent The \textsc{raw-PQC} policies are then a perfect fit for these algorithms as we can show that:

\begin{lemma}\label{lem:raw-PQC-bound}
Any \textsc{raw-PQC} policy as defined in Def.~\ref{def:raw-PQC} satisfies $D\leq 1$.
\end{lemma}
See Appendix \ref{sec:raw-PQC-bound} for a proof.

\begin{corollary}
Any \textsc{raw-PQC} policy as defined in Def.~\ref{def:raw-PQC} can benefit from a full quadratic speed-up from quantum numerical gradient estimation. 
\end{corollary}

\vspace{1em}

\section{Analytical gradient estimation}

We obtain our analytical policy gradient algorithm by applying the quantum multivariate Monte Carlo algorithm of Sec.~\ref{sec:quantum-MC-sub} to the formulation of the gradient given by the policy gradient theorem (see Sec.~\ref{sec:analytical-grad-prelim}). The random variable in this formulation
\begin{equation}
	X(\tau) = \sum_{t=0}^{T-1}\grad\log{\policy(a_t|s_t)} R(\tau)
\end{equation}
can easily be bounded in $\ell_p$-norm given an upper bound on the return $R(\tau)$ and the $\ell_p$-norm of the gradient of the log-policy:
\begin{equation}\label{eq:Bp}
	B_p =  \max_{s\in\S, a\in\A}\norm{\grad\log\policy(a|s)}_p.
\end{equation}
With this notation we can show the following Theorem:

\begin{thm}{Analytical policy gradient algorithm}{}
Let $\policy$ be a policy parametrized by a vector $\params\in\R^d$, that can be used to interact with a quantum-accessible MDP $\M=(\S,\A,P,R,\Rmax,T,\gamma)$, and such that $\policy$ has a bounded smoothness parameter $B_p$ for some $p\geq1$, defined in Eq.~(\ref{eq:Bp}). Call $\xi(p) = \max\{0,\frac{1}{2}-\frac{1}{p}\}$. The gradient of the value function corresponding to this policy, $\grad\valuefct(s_0)$, can be evaluated to error $\varepsilon$ in $\ell_\infty$-norm, using
\begin{equation}
	\widetilde{\O}\left( d^{\xi(p)}\frac{B_pT\Rmax}{\varepsilon(1-\gamma)}\right)
\end{equation}
length-$T$ episodes of interaction with the environment using a quantum analytical gradient estimator, while a classical analytical gradient estimator needs
\begin{equation}
	\widetilde{\O}\left( \left(\frac{B_pT\Rmax}{\varepsilon(1-\gamma)}\right)^2\right)
\end{equation}
length-$T$ episodes of interaction with the environment.\footnote{Note that the classical estimator still has a logarithmic dependence in $d$, hidden in the $\widetilde{\O}$ notation.} Notably, for $p\in[1,2]$, we get a full quadratic speed-up in the quantum setting.
\end{thm}

\begin{proof}
We apply Theorem \ref{thm:MVMC} for the random variable $X(\tau) = \sum_{t=0}^{T-1}\grad\log{\policy(a_t|s_t)} \sum_{t'=0}^{T-1} \gamma^{t'} r_{t'}$ distributed according to $P_{\params}(\tau) = \prod_{t=0}^{T-1}\policy(a_t|s_t)P(s_{t+1}|s_t,a_t)$.

To construct the appropriate quantum access to $X(\tau)$ (see Def.~\ref{def:quantum-exp}), we use the unitary $U_{P(\tau)}$ defined in Lemma \ref{lem:U-traj} to implement $U_P$, and implement the binary oracle $\B_X$ using the unitary $U_{R(\tau)}$ defined in Lemma \ref{lem:U-ret} along with a simulated classical circuit that multiplies the returns $R(\tau) = \sum_{t'=0}^{T-1} \gamma^{t'} r_{t'}$ with $\sum_{t=0}^{T-1}\grad\log{\policy(a_t|s_t)}$.

From Lemma \ref{lem:bound-valuefct}, we get the bound $\norm{X(\tau)}_p \leq \frac{TB_p\Rmax}{1-\gamma}$, which we use as the bound $B$ in Theorem \ref{thm:MVMC}, resulting in the stated quantum query complexity.\\
	
The classical complexity derives directly from Lemma \ref{lem:classical-MVMC} by noting that $\norm{X(\tau)}_\infty \leq \norm{X(\tau)}_p$ for any $p\geq 1$, and that sampling a trajectory $\tau$ (to compute a sample of $X(\tau)$) requires $1$ episode of interaction with the environment. 
\end{proof}

Note that the total query complexity of the quantum and classical analytical gradient estimators, in terms of the number of calls to $\P$ and $\Rwd$, is $\widetilde{\O}\left( d^{\xi(p)} \frac{B_pT^2\Rmax}{\varepsilon(1-\gamma)}\right)$ and $\widetilde{\O}\left( \frac{B_p^2T^3\Rmax^2}{\varepsilon^2(1-\gamma)^2}\right)$, respectively.

\noindent The $\textsc{softmax}_1$\textsc{-PQC} policies are then a perfect fit for these algorithms as we can show that:

\begin{lemma}\label{lem:softmax1-PQC-bound}
Any $\textsc{softmax}_1$\textsc{-PQC} policy as defined in Def.~\ref{def:softmax1-PQC} satisfies $B_1\leq 2$.
\end{lemma}

See Appendix \ref{sec:softmax1-PQC-bound} for a proof.

\begin{corollary}
Any $\textsc{softmax}_1$\textsc{-PQC} policy as defined in Def.~\ref{def:softmax1-PQC} can benefit from a full quadratic speed-up from quantum analytical gradient estimation. 
\end{corollary}

\section{Discussion}

In this work, we design quantum algorithms to train parametrized policies in quantum-accessible environments. These algorithms can provide up to quadratic speed-ups in the number of interactions needed to evaluate the parameter updates of these policies, provided the environments allow for the appropriate quantum access. Their sample complexity is mostly governed by the number of parameters $d$ of the policy, as well as the smoothness parameters $D$ and $B_p$, depending on whether the numerical or analytical gradient estimation is used. These two smoothness parameters are hard to relate to each other in general, making the performances of these two algorithms hard to compare. Nonetheless, we show that quantum policies previously studied in the literature are smooth with respect to each of these parameters (i.e., with $D$ or $B_1$ in $\O(1)$), which allows them to benefit from a full quadratic speed-up in sample complexity.

We note that in our results we only obtain quadratic speed-ups over specific classical algorithms that exploit the same smoothness conditions as our quantum algorithms. In order to strengthen these results, one would ideally prove matching lower bounds for the classical complexity of this task. We leave as an open question whether known classical lower bounds \cite{alabdulkareem21,LM19j} can be adapted to policy gradient evaluation.

In the analysis of the smoothness of the value function in Appendix \ref{sec:Gevrey-value-fct} (specifically around Eq.~(\ref{eq:gevrey-contrived-ub})), we end up bounding its derivatives $\partial_{\alphas}V^{(t)}_{\policy}(s)$ using a loose upper bound, especially in the regime where the order $k = \abs{\alphas}$ of the derivation is small. The reason for this loose bound is that we need to cast it as a Gevrey condition in order to apply the numerical gradient algorithms of Refs.~\cite{gilyen19,cornelissen19}. We conjecture that a modification of the construction in \cite{gilyen19,cornelissen19} may be possible such as to gain an improvement by a factor of $T$ in the sample complexity of our numerical gradient algorithm, and such that the resulting scaling in $T$ would match that of our analytical gradient estimation algorithm. Side-stepping the Gevrey-formulation of the bound would also remove the need for the condition $\gamma T \geq 2$ that we enforce in the MDP (which is in any case not a very limiting condition, as MDPs of interest usually have a large horizon $T$ and a discount factor $\gamma$ close to $1$).

\section*{Acknowledgments}

SJ acknowledges support from the Austrian Science Fund (FWF) through the projects DK-ALM:W1259-N27 and SFB BeyondC F7102. SJ also acknowledges the Austrian Academy of Sciences as a recipient of the DOC Fellowship. MO was supported by an NWO Vidi grant (Project No. VI.Vidi.192.109). This work was in part supported by the Dutch Research Council (NWO/OCW), as part of the Quantum Software Consortium programme (project number 024.003.037).

\bibliographystyle{unsrt}
\bibliography{references}

\begin{thebibliography}{10}

\bibitem{chia22}
Nai-Hui Chia, Andr{\'a}s~Pal Gily{\'e}n, Tongyang Li, Han-Hsuan Lin, Ewin Tang,
  and Chunhao Wang.
\newblock Sampling-based sublinear low-rank matrix arithmetic framework for
  dequantizing quantum machine learning.
\newblock {\em Journal of the ACM}, 69(5):1--72, 2022.

\bibitem{dunjko17}
Vedran Dunjko, Jacob~M Taylor, and Hans~J Briegel.
\newblock Advances in quantum reinforcement learning.
\newblock In {\em 2017 IEEE International Conference on Systems, Man, and
  Cybernetics (SMC)}, pages 282--287. IEEE, 2017.

\bibitem{grover96}
Lov~K Grover.
\newblock A fast quantum mechanical algorithm for database search.
\newblock In {\em Proceedings of the twenty-eighth annual ACM symposium on
  Theory of computing}, pages 212--219, 1996.

\bibitem{dunjko16}
Vedran Dunjko, Jacob~M Taylor, and Hans~J Briegel.
\newblock Quantum-enhanced machine learning.
\newblock {\em Physical review letters}, 117(13):130501, 2016.

\bibitem{saggio21}
Valeria Saggio, Beate~E Asenbeck, Arne Hamann, Teodor Str{\"o}mberg, Peter
  Schiansky, Vedran Dunjko, Nicolai Friis, Nicholas~C Harris, Michael Hochberg,
  Dirk Englund, et~al.
\newblock Experimental quantum speed-up in reinforcement learning agents.
\newblock {\em Nature}, 591(7849):229--233, 2021.

\bibitem{hamann21}
Arne Hamann, Vedran Dunjko, and Sabine W{\"o}lk.
\newblock Quantum-accessible reinforcement learning beyond strictly epochal
  environments.
\newblock {\em Quantum Machine Intelligence}, 3(2):1--18, 2021.

\bibitem{wang21}
Daochen Wang, Xuchen You, Tongyang Li, and Andrew~M Childs.
\newblock Quantum exploration algorithms for multi-armed bandits.
\newblock In {\em Proceedings of the AAAI Conference on Artificial
  Intelligence}, volume~35, pages 10102--10110, 2021.

\bibitem{wang21b}
Daochen Wang, Aarthi Sundaram, Robin Kothari, Ashish Kapoor, and Martin
  Roetteler.
\newblock Quantum algorithms for reinforcement learning with a generative
  model.
\newblock In {\em International Conference on Machine Learning}, pages
  10916--10926. PMLR, 2021.

\bibitem{ronagh19}
Pooya Ronagh.
\newblock The problem of dynamic programming on a quantum computer.
\newblock {\em arXiv:1906.02229}, 2019.

\bibitem{cherrat22}
El~Amine Cherrat, Iordanis Kerenidis, and Anupam Prakash.
\newblock Quantum reinforcement learning via policy iteration.
\newblock {\em arXiv:2203.01889}, 2022.

\bibitem{dunjko17b}
Vedran Dunjko, Yi-Kai Liu, Xingyao Wu, and Jacob~M Taylor.
\newblock Exponential improvements for quantum-accessible reinforcement
  learning.
\newblock {\em arXiv:1710.11160}, 2017.

\bibitem{silver17}
David Silver, Julian Schrittwieser, Karen Simonyan, Ioannis Antonoglou, Aja
  Huang, Arthur Guez, Thomas Hubert, Lucas Baker, Matthew Lai, Adrian Bolton,
  et~al.
\newblock Mastering the game of go without human knowledge.
\newblock {\em Nature}, 550(7676):354, 2017.

\bibitem{mirowski18}
Piotr Mirowski, Matt Grimes, Mateusz Malinowski, Karl~Moritz Hermann, Keith
  Anderson, Denis Teplyashin, Karen Simonyan, Andrew Zisserman, Raia Hadsell,
  et~al.
\newblock Learning to navigate in cities without a map.
\newblock {\em Advances in Neural Information Processing Systems},
  31:2419--2430, 2018.

\bibitem{mnih15}
Volodymyr Mnih, Koray Kavukcuoglu, David Silver, Andrei~A Rusu, Joel Veness,
  Marc~G Bellemare, Alex Graves, Martin Riedmiller, Andreas~K Fidjeland, Georg
  Ostrovski, et~al.
\newblock Human-level control through deep reinforcement learning.
\newblock {\em Nature}, 518(7540):529, 2015.

\bibitem{williams92}
Ronald~J Williams.
\newblock Simple statistical gradient-following algorithms for connectionist
  reinforcement learning.
\newblock {\em Machine learning}, 8(3-4):229--256, 1992.

\bibitem{sutton00}
Richard~S Sutton, David~A McAllester, Satinder~P Singh, and Yishay Mansour.
\newblock Policy gradient methods for reinforcement learning with function
  approximation.
\newblock In {\em Advances in neural information processing systems}, pages
  1057--1063, 2000.

\bibitem{kohl04}
Nate Kohl and Peter Stone.
\newblock Policy gradient reinforcement learning for fast quadrupedal
  locomotion.
\newblock In {\em IEEE International Conference on Robotics and Automation,
  2004. Proceedings. ICRA'04. 2004}, volume~3, pages 2619--2624. IEEE, 2004.

\bibitem{jerbi21}
Sofiene Jerbi, Casper Gyurik, Simon Marshall, Hans Briegel, and Vedran Dunjko.
\newblock Parametrized quantum policies for reinforcement learning.
\newblock {\em Advances in Neural Information Processing Systems}, 34, 2021.

\bibitem{sequeira22}
Andr{\'e} Sequeira, Luis~Paulo Santos, and Lu{\'\i}s~Soares Barbosa.
\newblock Variational quantum policy gradients with an application to quantum
  control.
\newblock {\em arXiv:2203.10591}, 2022.

\bibitem{chen22}
Samuel Yen-Chi Chen, Chih-Min Huang, Chia-Wei Hsing, Hsi-Sheng Goan, and
  Ying-Jer Kao.
\newblock Variational quantum reinforcement learning via evolutionary
  optimization.
\newblock {\em Machine Learning: Science and Technology}, 3(1):015025, 2022.

\bibitem{chen20}
Samuel Yen-Chi Chen, Chao-Han~Huck Yang, Jun Qi, Pin-Yu Chen, Xiaoli Ma, and
  Hsi-Sheng Goan.
\newblock Variational quantum circuits for deep reinforcement learning.
\newblock {\em IEEE Access}, 8:141007--141024, 2020.

\bibitem{lockwood20}
Owen Lockwood and Mei Si.
\newblock Reinforcement learning with quantum variational circuit.
\newblock In {\em Proceedings of the AAAI Conference on Artificial Intelligence
  and Interactive Digital Entertainment}, volume~16, pages 245--251, 2020.

\bibitem{wu20}
Shaojun Wu, Shan Jin, Dingding Wen, and Xiaoting Wang.
\newblock Quantum reinforcement learning in continuous action space.
\newblock {\em arXiv:2012.10711}, 2020.

\bibitem{skolik21}
Andrea Skolik, Sofiene Jerbi, and Vedran Dunjko.
\newblock Quantum agents in the gym: a variational quantum algorithm for deep
  q-learning.
\newblock {\em Quantum}, 6:720, 2022.

\bibitem{gilyen19}
Andr{\'a}s Gily{\'e}n, Srinivasan Arunachalam, and Nathan Wiebe.
\newblock Optimizing quantum optimization algorithms via faster quantum
  gradient computation.
\newblock In {\em Proceedings of the Thirtieth Annual ACM-SIAM Symposium on
  Discrete Algorithms}, pages 1425--1444. SIAM, 2019.

\bibitem{cornelissen19}
Arjan Cornelissen.
\newblock Quantum gradient estimation of gevrey functions.
\newblock {\em arXiv:1909.13528}, 2019.

\bibitem{cornelissen21}
Arjan Cornelissen and Sofiene Jerbi.
\newblock Quantum algorithms for multivariate monte carlo estimation.
\newblock {\em arXiv:2107.03410}, 2021.

\bibitem{cornelissen22}
Arjan Cornelissen, Yassine Hamoudi, and Sofiene Jerbi.
\newblock Near-optimal quantum algorithms for multivariate mean estimation.
\newblock In {\em Proceedings of the 54th Annual ACM SIGACT Symposium on Theory
  of Computing}, pages 33--43, 2022.

\bibitem{kakade03}
Sham~Machandranath Kakade.
\newblock {\em On the sample complexity of reinforcement learning}.
\newblock PhD thesis, UCL (University College London), 2003.

\bibitem{silver15}
David Silver.
\newblock Lectures on reinforcement learning.
\newblock \textsc{url:}~\url{https://www.davidsilver.uk/teaching/}, 2015.

\bibitem{sutton98}
Richard~S Sutton, Andrew~G Barto, et~al.
\newblock {\em Reinforcement learning: An introduction}.
\newblock 1998.

\bibitem{grover02}
Lov Grover and Terry Rudolph.
\newblock Creating superpositions that correspond to efficiently integrable
  probability distributions.
\newblock {\em quant-ph/0208112}, 2002.

\bibitem{benedetti19}
Marcello Benedetti, Erika Lloyd, Stefan Sack, and Mattia Fiorentini.
\newblock Parameterized quantum circuits as machine learning models.
\newblock {\em Quantum Science and Technology}, 4(4):043001, 2019.

\bibitem{jordan05}
Stephen~P Jordan.
\newblock Fast quantum algorithm for numerical gradient estimation.
\newblock {\em Physical review letters}, 95(5):050501, 2005.

\bibitem{montanaro15}
Ashley Montanaro.
\newblock Quantum speedup of monte carlo methods.
\newblock {\em Proceedings of the Royal Society A: Mathematical, Physical and
  Engineering Sciences}, 471(2181):20150301, 2015.

\bibitem{grover98}
Lov~K Grover.
\newblock A framework for fast quantum mechanical algorithms.
\newblock In {\em Proceedings of the thirtieth annual ACM symposium on Theory
  of computing}, pages 53--62, 1998.

\bibitem{hamoudi21}
Yassine Hamoudi.
\newblock Quantum sub-gaussian mean estimator.
\newblock In {\em 29th Annual European Symposium on Algorithms (ESA 2021)}.
  Schloss Dagstuhl-Leibniz-Zentrum f{\"u}r Informatik, 2021.

\bibitem{huggins21}
William~J Huggins, Kianna Wan, Jarrod McClean, Thomas~E O'Brien, Nathan Wiebe,
  and Ryan Babbush.
\newblock Nearly optimal quantum algorithm for estimating multiple expectation
  values.
\newblock {\em arXiv:2111.09283}, 2021.

\bibitem{alabdulkareem21}
Abdulrahman Alabdulkareem and Jean Honorio.
\newblock Information-theoretic lower bounds for zero-order stochastic gradient
  estimation.
\newblock In {\em 2021 IEEE International Symposium on Information Theory
  (ISIT)}, pages 2316--2321. IEEE, 2021.

\bibitem{LM19j}
G.~Lugosi and S.~Mendelson.
\newblock Mean estimation and regression under heavy-tailed distributions: {A}
  survey.
\newblock {\em Foundations of Computational Mathematics}, 19(5):1145--1190,
  2019.

\bibitem{schuld19}
Maria Schuld, Ville Bergholm, Christian Gogolin, Josh Izaac, and Nathan
  Killoran.
\newblock Evaluating analytic gradients on quantum hardware.
\newblock {\em Physical Review A}, 99(3):032331, 2019.

\bibitem{cerezo21b}
Marco Cerezo and Patrick~J Coles.
\newblock Higher order derivatives of quantum neural networks with barren
  plateaus.
\newblock {\em Quantum Science and Technology}, 6(3):035006, 2021.

\end{thebibliography}

\appendix

\section{Simple derivation of the policy gradient theorem\label{appdx:derivation-pgt}}

\begin{theorem}[Policy gradient theorem \cite{sutton00}]\label{thm:pgt}
	Given a policy $\policy$ that generates trajectories $\tau = (s_0, a_0, r_0, s_1, \ldots)$ in a reinforcement learning environment with time horizon $T\in\N\cup\{\infty\}$, the gradient of the value function $\valuefct$ with respect to $\params$ is given by
	\begin{equation}\label{eq:pgt}
		\grad\valuefct (s_0) = \mathbb{E}_{\tau}  \left[ \sum_{t=0}^{T-1}\grad\log{\policy(a_t|s_t)} \sum_{t'=0}^{T-1} \gamma^{t'} r_{t'} \right].
	\end{equation}
\end{theorem}
\begin{proof}
Call $R(\tau) = \sum_{t=0}^{T-1} \gamma^{t} r_{t}$ the return of a trajectory $\tau$, and $P_{\params}(\tau) = \prod_{t=0}^{T-1}\policy(a_t|s_t)P_E(s_{t+1}|s_t,a_t)$ the probability of this trajectory, where $P_E$ describes the unknown dynamics of the environment.

Then, we can write the value function as
\begin{equation}
	\valuefct (s_0)  = \sum_{\tau} P_{\params}(\tau) R(\tau)
\end{equation}
and its gradient as
\begin{align}
	\grad\valuefct (s_0) &= \sum_{\tau} \grad P_{\params}(\tau) R(\tau)\\
	&= \sum_{\tau} P_{\params}(\tau) \frac{\grad P_{\params}(\tau)}{P_{\params}(\tau)} R(\tau)\\
	&= \sum_{\tau} P_{\params}(\tau) \grad \log(P_{\params}(\tau)) R(\tau)\\
	&= \sum_{\tau} P_{\params}(\tau) \sum_{t=0}^{T-1} \grad\log(\policy(a_t|s_t)) R(\tau)\\
	&= \mathbb{E}_{\tau} \left[ \sum_{t=0}^{T-1} \grad\log(\policy(a_t|s_t)) R(\tau) \right]
\end{align}
where we have artificially divided and multiplied each term by $P_{\params}(\tau)$ in the second line, and used the independence on $\params$ of the environment dynamics $P_E(s_{t+1}|s_t,a_t)$ in the fourth line.
\end{proof}

\section{Lemmas concerning properties of MDPs}

\subsection{An upper bound on the value function}

\begin{lemma}\label{lem:bound-valuefct}
Consider an MDP $\M = (\S,\A,P,R,\Rmax,T,\gamma)$ as defined in Def.~\ref{def:MDP}. The value function $\valuefct(s_0) = \E\left[\sum_{t=0}^{T-1} \gamma^t r_t\right]$ of any policy $\policy$, evaluated on any initial state $s_0\in\S$ is upper bounded by
\begin{equation}
	\abs{\valuefct(s_0)} \leq \min\left\{T,\frac{1}{1-\gamma}\right\}\Rmax.
\end{equation}
\end{lemma}
\begin{proof}
We have, by definition of the MDP, $r_t \leq \Rmax$, which implies:
\begin{equation}
	\abs{\sum_{t=0}^{T-1} \gamma^t r_t} \leq \sum_{t=0}^{T-1} \gamma^t \abs{r_t} \leq \sum_{t=0}^{T-1} \gamma^t \Rmax \leq
	\begin{cases}
	\frac{\Rmax}{1-\gamma} & \text{if }\gamma<1\\
	T\Rmax & \text{always}
	\end{cases}
\end{equation}
which also holds in expectation value over all trajectories of length $T$.
\end{proof}

\subsection{The effective time horizon of an MDP}
\begin{lemma}\label{lem:effective-horizon}
Consider an MDP $\M = (\S,\A,P,R,\Rmax,T,\gamma)$ as defined in Def.~\ref{def:MDP}, with an infinite horizon $T=\infty, \gamma < 1$ and a value function $\valuefct$. The finite-horizon MDP $\M' = (\S,\A,P,R,\Rmax,T^{*},\gamma)$, where
\begin{equation}
	T^{*} = \left\lceil\frac{\log\left(\frac{\varepsilon(1-\gamma)}{\Rmax}\right)}{\log(\gamma)}\right\rceil = \widetilde{\O}\left(\frac{1}{1-\gamma}\right)
\end{equation}
has a value function $\valuefct'$ that satisfies
\begin{equation}
	\abs{\valuefct(s_0)-\valuefct'(s_0)} \leq \varepsilon
\end{equation}
for any initial state $s_0\in\S$ and any policy $\policy$.
\end{lemma}
\begin{proof}
\begin{align}
	\abs{\valuefct(s_0)-\valuefct'(s_0)} &= \abs{\E\left[\sum_{t=0}^{\infty} \gamma^t r_t\right]-\E\left[\sum_{t=0}^{T^{*}-1} \gamma^t r_t\right]}\\
	&=  \abs{\E\left[\sum_{t=T^{*}}^{\infty} \gamma^t r_t\right]}\\
	&\leq \gamma^{T^{*}}\frac{\Rmax}{1-\gamma}\\
	&\leq \frac{\varepsilon(1-\gamma)}{\Rmax}\frac{\Rmax}{1-\gamma} = \varepsilon.
\end{align}
\end{proof}

Because of this lemma, we always assume the time horizon $T$ of an MDP to be in $\widetilde{\O}\left(\frac{1}{1-\gamma}\right)$.

\section{Complexity of a classical MVMC algorithm\label{sec:classical-complexity-MVMC}}

\begin{lemma}[Classical multivariate Monte Carlo estimation]\label{lem:classical-MVMC}
Let $X$ be a $d$-dimensional bounded random variable such that $\norm{X}_\infty \leq B$. Given sampling access to $X$, $\varepsilon,\delta >0$, there exists a classical multivariate mean estimator that returns an $\varepsilon$-precise estimate of $\E[X]$ in $\ell_\infty$-norm with success probability at least $1-\delta$ using
\begin{equation}
	\widetilde{\O}\left(\left(\frac{B}{\varepsilon}\right)^2\right)
\end{equation}
samples of $X$.
\end{lemma}

\begin{proof}
Consider the following algorithm:\\
\begin{enumerate}
	\item Collect $N = \left\lceil \frac{2B^2}{\varepsilon^2} \log\left(\frac{2d}{\delta}\right)\right\rceil$ samples of $X$: $\left\{\bm{x}^{(i)}=(x^{(i)}_1, \ldots, x^{(i)}_d)\right\}_{1\leq i \leq N}$.
	\item Compute the $d$ coordinate-wise averages $\widehat{x}_j = \frac{1}{N}\sum_{i=1}^{N} x^{(i)}_j$ and use $\widehat{\bm{x}} = (\widehat{x}_1, \ldots, \widehat{x}_d)$ as an estimate.
\end{enumerate}
Now consider the probability of failure of this algorithm, i.e., that at least one of the estimates is more than $\varepsilon$ away from its expected value:
\begin{align*}
\Pbb\left( \bigvee_{j \in [d]} \abs{\widehat{x}_j - \mathbb{E}[x_j]} \geq \varepsilon \right) &\leq \sum_{j=1}^d \Pbb\left(\abs{\widehat{x}_j - \mathbb{E}[x_j]} \geq \varepsilon \right) & \#\textrm{ union bound}\\
&\leq d \times \max_{j \in [d]} \P\left(\abs{\widehat{x}_j - \mathbb{E}[x_j]} \geq \varepsilon \right) &\\
&\leq 2d\exp\left( -\frac{2N^2\varepsilon^2}{4NB^2} \right)  & \#\textrm{ Hoeffding's bound and bound on $x_j$}\\
&\leq \delta. & \#\textrm{ definition of }N
\end{align*}
Hence, for arbitrary $\varepsilon$ and $\delta$, the $d$ expectations can be estimated to error $\varepsilon$ in the $\ell_\infty$-norm with success probability $1-\delta$ using $N=\mathcal{O} \left( \frac{B^2}{\varepsilon^2} \log(\frac{d}{\delta})\right)$ samples of $X$.
\end{proof}

\section{Proof of Lemma \ref{lem:raw-PQC-bound}\label{sec:raw-PQC-bound}}

\begin{repeatlem}{lem:raw-PQC-bound}
Any \textsc{raw-PQC} policy as defined in Def.~\ref{def:raw-PQC} satisfies $D\leq 1$.
\end{repeatlem}

\begin{proof}
Given a \textsc{raw-PQC} policy $\policy$ as defined in Def.~\ref{def:raw-PQC}, we seek to bound the following quantity:
\begin{equation}
	D = \max_{k\in\N^*} (D_k)^{1/k}
\end{equation}
where
\begin{equation}\label{eq:dk}
	D_k = \max_{s\in\S, \alphas\in[d]^k} \sum_{a\in\A}\abs{\partial_{\alphas} \policy(a|s)}.
\end{equation}

Gradients of this PQC policy can be evaluated using the parameter-shift rule \cite{schuld19}:
\begin{equation}
	\partial_i \policy(a|s) = \partial_i \expval{P_a}_{s,\params} = \frac{\expval{P_a}_{s,\params+\frac{\pi}{2}e_i} - \expval{P_a}_{s,\params-\frac{\pi}{2}e_i}}{2}
\end{equation}
which can easily be generalized to higher-order derivatives \cite{cerezo21b}:
\begin{equation}\label{eq:multi-psr}
	\partial_{\alphas} \policy(a|s) = \frac{1}{2^k} \sum_{\omegas} c_{\omegas} \expval{P_a}_{s,\params + \omegas}
\end{equation}
for $\alphas \in [d]^k, \omegas \in \{0, \pm\frac{\pi}{2}, \pm\pi, \pm\frac{3\pi}{2}\}^d$, and $c_{\omegas} \in \Z$ such that $\sum_{\omegas} \abs{c_{\omegas}} = 2^k$.

Now, by combining Eq.~(\ref{eq:dk}) and (\ref{eq:multi-psr}), we get:
\begin{align}
D_k &= \max_{s\in\S, \alphas\in[d]^k} \sum_{a\in\A}\abs{\frac{1}{2^k}\sum_{\omegas} c_{\omegas} \expval{P_a}_{s,\params + \omegas}}\\
	&\leq \max_{s\in\S, \alphas\in[d]^k} \frac{1}{2^k} \sum_{a\in\A}\sum_{\omegas} \abs{c_{\omegas}} \abs{\expval{P_a}_{s,\params + \omegas}}\\
	&= \max_{s\in\S, \alphas\in[d]^k} \frac{1}{2^k} \sum_{\omegas} \abs{c_{\omegas}} \sum_{a\in\A} \abs{\expval{P_a}_{s,\params + \omegas}} = 1.
\end{align}
where in the last line we used $\sum_{a} P_a = I$ in the definition of the \textsc{raw-PQC} policy and $\sum_{\omegas} \abs{c_{\omegas}} = 2^k$.\\
Since this bound is valid for all $k\in\N^*$, then also $D\leq1$.
\end{proof}

\section{Proof of Lemma \ref{lem:softmax1-PQC-bound}\label{sec:softmax1-PQC-bound}}

\begin{repeatlem}{lem:softmax1-PQC-bound}
Any $\textsc{softmax}_1$\textsc{-PQC} policy as defined in Def.~\ref{def:softmax1-PQC} satisfies $B_1\leq 2$.
\end{repeatlem}

\begin{proof}
Given a $\textsc{softmax}_1$\textsc{-PQC} policy $\policy$ as defined in Def.~\ref{def:softmax1-PQC}, we seek to bound the following quantity:
\begin{equation}
	B_1 =  \max_{s\in\S, a\in\A}\norm{\grad\log\policy(a|s)}_1.
\end{equation}

From the definition of this policy, we have:
\begin{equation}
	\expval{O_a}_{s,\params} = \ev{\sum_i w_{a,i} P_{a,i}}{\psi_{s}}
\end{equation}
such that $\sum_i P_{a,i}=I$ and $P_{a,i}P_{a,i'} = \delta_{i,i'}P_{a,i}$, $\forall a\in\A$. This implies that
\begin{equation}
	\partial_{w_{a',i}}\expval{O_a}_{s,\params} = \delta_{a,a'}\ev{P_{a',i}}{\psi_{s}} = \delta_{a,a'}\expval{P_{a',i}}_{s}.
\end{equation}
Since this is a \textsc{softmax-PQC}, it follows from Lemma 1 in \cite{jerbi21} that:
\begin{align}
	\partial_{w_{a',i}}\log\policy(a|s) &= \partial_{w_{a',i}}\expval{O_a}_{s,\params} - \sum_{a''\in\A} \policy(a''|s) \partial_{w_{a',i}}\expval{O_{a''}}_{s,\params}\\
	&= \delta_{a,a'} \expval{P_{a',i}}_{s} - \policy(a'|s) \expval{P_{a',i}}_{s}.
\end{align}
Therefore,
\begin{align}
	\norm{\grad\log\policy(a|s)}_1  &= \sum_{a',i} \abs{\partial_{w_{a',i}}\log\policy(a|s)}\\
	&\leq \sum_{a',i} \left[ \abs{\delta_{a,a'} \expval{P_{a',i}}_{s}} + \abs{\policy(a'|s) \expval{P_{a',i}}_{s}} \right]\\
	&\leq \sum_i \expval{P_{a,i}}_{s} + \sum_{a',i} \policy(a'|s) \expval{P_{a',i}}_{s}\\
	&\leq 1 + \max_{a'} \sum_i \expval{P_{a',i}}_{s}\\
	&\leq 2
\end{align}
where we made use of the triangle inequality in the first inequality, the positivity of $\expval{P_{a,i}}_{s}$ and $\policy(a'|s)$ in the second inequality, and the normalization constraint of $\{P_{a,i}\}_i$ in the third and fourth inequalities.
\end{proof}

\section{Gevrey condition of value functions\label{sec:Gevrey-value-fct}}

In this section, we investigate the smoothness of the value function, in terms of the smoothness of the policy. More precisely, we prove the following lemma:

\begin{lemma}\label{lem:gevrey-valuefct}
Let $\policy$ be a parametrized policy with a bounded smoothness parameter $D$, defined in Eq.~(\ref{eq:D}). Let $\M = (\S,\A,P,R,\Rmax,T,\gamma)$ be an MDP as defined in Def.~\ref{def:MDP} with $T\gamma \geq 2$. Then the value function $\valuefct(s_0)$ associated to the policy $\policy$ in $\M$, as a function of the policy parameters $\params$, satisfies the Gevrey conditions of Def.~\ref{def:gevrey_functions} for $\sigma = 0, M=\frac{4\Rmax}{1-\gamma}$ and $c=DT^2$.
\end{lemma}

As a first step, we observe that we can use the Markovian nature of an MDP to describe the value function as the limit of a sequence of improving approximations, by iteratively increasing the time horizon at which we evaluate the MDP. More precisely, we define inductively, for all states $s \in \S$ and time horizons $t \geq 0$,
\[V^{(t+1)}_{\policy}(s) = \sum_{a \in \A} \policy(a|s) \left[R(s,a) + \gamma \sum_{s' \in \S} P(s'|s,a) V_{\policy}^{(t)}(s')\right],\]
where for the induction basis, we use $V^{(0)}_{\policy}(s) = 0$, for all states $s \in \S$. We easily check that the value function at time horizon $T \in \N \cup \{\infty\}$ of an MDP, $V_{\policy}(s)$, is indeed given by letting $t$ go to $T$ in the above definition.

This recursive definition of approximations to the value function provides us with a convenient handle on its derivatives. In particular, for all integers $k,t > 0$ and sequences $\alphas \in [d]^k$, where $d$ is the number of parameters of $\params$, i.e., $\params \in \R^d$, we obtain that
\begin{equation}
    \label{eq:value-function-recursive-derivatives}
    \partial_{\alphas}\left[V^{(t+1)}_{\policy}(s) - V^{(t)}_{\policy}(s)\right] = \gamma \partial_{\alphas} \left[\sum_{a \in \A} \policy(a|s) \sum_{s' \in \S} P(s'|s,a) (V^{(t)}_{\policy}(s') - V^{(t-1)}_{\policy}(s'))\right].
\end{equation}
Since the value function with time horizon $t = 0$ vanishes, we can express the partial derivatives at any given time horizon $t$ as the telescoping sum
\begin{align*}
    \partial_{\alphas}V^{(t)}_{\policy}(s) = \sum_{t'=0}^{t-1} \partial_{\alphas}\left[V^{(t'+1)}_{\policy}(s) - V^{(t')}_{\policy}(s)\right].
\end{align*}
The main idea of this section is to expand the expression on the right-hand side in the above equation, using the recursive characterization provided in Eq.~(\ref{eq:value-function-recursive-derivatives}).

We start by defining some shorthand notation:
\begin{definition}\label{def:g-U}
    Let $\M=(\S,\A,P,R,\Rmax,T,\gamma)$ be an MDP, and $\policy$ be a policy parametrized by $\params \in \R^d$. Let $V^{(t)}_{\policy}$ be its value function with horizon $t > 0$, and for all $k, t > 0$, let
    \[g(k,t) = \max_{s \in \S, \alphas \in [d]^k} \left|\partial_{\alphas} \left[ V^{(t+1)}_{\policy}(s) - V^{(t)}_{\policy}(s)\right]\right|, \qquad \text{and} \qquad U(k,t) = \sum_{t'=0}^{t-1} g(k,t').\]
\end{definition}

We observe that
\begin{equation}\label{eq:Ukt-bound}
|\partial_{\alphas}V^{(t)}_{\policy}(s)| \leq \sum_{t'=0}^{t-1} \left|\partial_{\alphas}\left[V^{(t'+1)}_{\policy}(s) - \partial_{\alphas}V^{(t')}_{\policy}(s)\right]\right| \leq \sum_{t'=0}^{t-1} g(k,t') = U(k,t),
\end{equation}
and hence to bound the smoothness of (the approximations to) the value function, it suffices to find a good upper bound on $U(k,t)$. The previous definition already foreshadows that the resulting expression explicitly depends on the smoothness of the policy through the parameter $D$.

In order to upper bound $U(k,t)$, we first find an expression that upper bounds $g(k,t)$, which is the objective of the following lemma.

\begin{lemma}\label{lem:gevrey-lem1}
	Let $\M=(\S,\A,P,R,\Rmax,T,\gamma)$ be an MDP, and $\policy$ be a policy parametrized by $\params \in \R^d$. Let $V^{(t)}_{\policy}$ be its value function with horizon $t > 0$. For all $k \in \N$, let $\Lambda_k$ be the set of all partitions of $k$, where every partition $\lambda \in \Lambda_k$ is a multiset of positive integers that sums to $k$. By $\{\lambda\}$, we denote the set of elements in $\lambda$, i.e., without repetition. We let $\#\ell(\lambda)$ be the number of occurrences of $\ell$ in the multiset $\lambda$, and let $\#\lambda = \{\#\ell(\lambda) : \ell \in \{\lambda\}\}$ be the multiset of occurrences in $\lambda$. For all non-negative integers $k,t$, we have
	\[g(k,t) \leq \gamma^t\Rmax \cdot \sum_{\lambda \in \Lambda_k} \binom{k}{\lambda} \binom{|\lambda|}{\#\lambda} \binom{t+1}{|\lambda|} \prod_{\ell \in \lambda} D_{\ell}.\]
\end{lemma}

\begin{proof}
	We give a combinatorial argument. To that end, let $k,t \geq 0$ be integers, and let $\alphas \in [d]^k$ be a finite sequence of indices with respect to which we want to compute the partial derivative of $V^{(t)}_{\policy}$. The main idea is to apply the product rule to the expression on the right-hand side of Eq.~(\ref{eq:value-function-recursive-derivatives}).
	
	In particular, by repeatedly substituting the right-hand side of Eq.~(\ref{eq:value-function-recursive-derivatives}) into itself, we obtain that there are $t+1$ probabilities $\policy(a|s)$ to which we can associate any given index of $\alphas$. Thus, we count the number of occurrences where the distribution of indices in $\alphas$ across the $t+1$ different factors forms the partition $\lambda \in \Lambda_k$. We call this number $c_{\lambda}$, and we indeed observe that all these terms are upper bounded by $\prod_{\ell \in \lambda} D_{\ell}$, which means that it remains to prove that
	\[c_{\lambda} = \binom{k}{\lambda} \binom{|\lambda|}{\#\lambda} \binom{t+1}{|\lambda|}.\]
	
	Observe that we must first choose which factors to assign any derivative to at all, which can be done in $\binom{t+1}{|\lambda|}$ ways. Then, we must decide how many derivatives we are going to assign to each of the selected factors, which can be done in $\binom{|\lambda|}{\#\lambda}$ ways. Finally, we must distribute the $k$ derivatives among the groups, which can be done in $\binom{k}{\lambda}$ ways. This completes the proof.
\end{proof}

Now that we have found an expression that upper bounds $g(k,t)$, we can use it to upper bound $U(k,t)$ as well. This is the objective of the following lemma.

\begin{lemma}
	Let $\M=(\S,\A,P,R,\Rmax,T,\gamma)$ be an MDP, and $\policy$ be a policy parametrized by $\params \in \R^d$. Let $V^{(t)}_{\policy}$ be its value function with horizon $t > 0$. For all non-negative integers $k,t$ such that $\gamma \geq 2/t$, we have
	\[U(k,t) \leq \frac{2|R|_{\max}}{1-\gamma} \cdot (\gamma Dt^2)^k.\]
\end{lemma}

\begin{proof}
	By plugging in the bound derived in Lemma~\ref{lem:gevrey-lem1}, we obtain directly that
	\begin{equation}
		\label{eq:gevrey-summation-equation}
		U(k,t) = \sum_{t'=0}^{t-1} g(k,t) \leq \sum_{t'=0}^{t-1} \gamma^{t'}|R|_{\max} \sum_{\lambda \in \Lambda_k} \binom{k}{\lambda} \binom{|\lambda|}{\#\lambda} \binom{t'+1}{|\lambda|} \prod_{\ell \in \lambda} D_{\ell}.
	\end{equation}
	First, for all $\lambda \in \Lambda_k$, we observe that the final product can be upper bounded as
	\[\prod_{\ell \in \lambda} D_{\ell} = \prod_{\ell \in \lambda} (D_{\ell}^{1/\ell})^{\ell} \leq \prod_{\ell \in \lambda} D^{\ell} = D^k.\]
	Next, we can swap the summations in Eq.~(\ref{eq:gevrey-summation-equation}), and after rewriting we obtain
	\begin{equation}
		\label{eq:gevrey-multinomial-thm-setup}
		U(k,t) \leq |R|_{\max}D^k \cdot \sum_{r=1}^k \sum_{\substack{k_1, \dots, k_r \in \N \\ k_1 + \cdots + k_r = k}} \binom{k}{k_1, \dots, k_r} r! \cdot \sum_{t'=0}^{t-1} \gamma^{t'} \binom{t'+1}{r}.
	\end{equation}
	We now focus on the final summation in the above expression. First, we observe that if $t < r$, then all the binomial coefficients evaluate to $0$, and therefore the summation as a whole vanishes as well. Thus, the only terms in the above expression that are non-zero are those where $r \leq t$, which means that we can change the upper limit of summation in the outermost summation to $\min(k,t)$. We can take at least $r$ factors of $\gamma$ out, and as such obtain
	\[\sum_{t'=0}^{t-1} \gamma^{t'} \binom{t'+1}{r} = \gamma^r \sum_{t'=0}^{t-r-1} \gamma^{t'} \binom{t'+r+1}{r} \leq \gamma^r \binom{t}{r} \sum_{t'=0}^{t-r-1} \gamma^{t'} \leq \frac{\gamma^rt^r}{(1-\gamma)r!}.\]
	Plugging this expression back into Eq.~(\ref{eq:gevrey-multinomial-thm-setup}) yields
	\begin{equation}
		\label{eq:gevrey-contrived-ub}
		U(k,t) \leq \frac{|R|_{\max}D^k}{1-\gamma} \cdot \sum_{r=1}^{\min(k,t)} (\gamma t)^r \sum_{\substack{k_1, \dots, k_r \in \N \\ k_1 + \cdots + k_r = k}} \binom{k}{k_1, \dots, k_r} = \frac{|R|_{\max}D^k}{1-\gamma} \cdot \sum_{r=1}^{\min(k,t)} (\gamma t)^r r^k.
	\end{equation}
	In the summation on the right-hand side, the last term is by far the biggest. We can show this crudely by observing that for all $a \geq 2$,
	\[\frac{1}{n^ka^n} \sum_{r=1}^n r^ka^r = \sum_{r=1}^n \left(\frac{r}{n}\right)^k a^{r-n} \leq \sum_{r=1}^n \left(\frac{1}{a}\right)^{n-r} \leq \sum_{r=0}^{n-1} \left(\frac12\right)^r \leq 2.\]
	Thus, by setting $n = \min(k,t)$, and $a = \gamma t$, we obtain that
	\[U(k,t) \leq \frac{2|R|_{\max}}{1-\gamma} \cdot (\gamma Dt^2)^k.\]
	This completes the proof.
\end{proof}

Lemma \ref{lem:gevrey-valuefct} then follows immediately from this lemma and Eq.~(\ref{eq:Ukt-bound}) for $t=T$.

\section{Classical complexity of numerical gradient estimation\label{sec:classical-num}}

In this Appendix, we analyze the complexity of a classical numerical gradient estimation algorithm that relies on the same smoothness conditions of the value function as the quantum algorithm. More precisely, we show the following lemma:

\begin{lemma}\label{lem:classical-num}
Let $\policy$ be a parametrized policy that can be used to interact with an MDP, and that has a bounded smoothness parameter $D$, defined in Eq.~(\ref{eq:D}). The gradient of the value function corresponding to this policy $\grad\valuefct(s_0)$ can be evaluated to error $\varepsilon$ in the $\ell_\infty$-norm, using
\begin{equation}
	\widetilde{\O}\left( d \left(\frac{DT^2\Rmax}{\varepsilon(1-\gamma)}\right)^2\right)
\end{equation}
length-$T$ episodes of interaction with the environment using a classical numerical gradient estimator.
\end{lemma}

To prove this lemma, we consider a central-difference method that, compared to a simple finite-difference method, can exploit more evaluations of a function $f$ and bounds on its higher-order derivatives to evaluate $f'(x)$ with higher precision. We perform an error analysis of this method and calculate its query complexity for functions $f$ that cannot be evaluated exactly but only through Monte Carlo estimation (such as value functions).

\subsection{Central difference numerical differentiation}

Suppose that we can evaluate a function $f:\mathbb{R}\rightarrow \mathbb{R}$ that is $k$ times differentiable at some point $x\in\mathbb{R}$, with $f^{(k-1)}$ continuous on some interval around $x$. For $\delta\in\mathbb{R}$ such that $x+\delta$ is in this interval, Taylor's theorem (with the Lagrange formulation of the remainder) gives us:
\begin{equation}
f\left(x+\delta\right) = f(x) + f'(x) \delta + \frac{f''(x)}{2!} \delta^2 + \ldots + \frac{f^{(k-1)}(x)}{(k-1)!} \delta^{k-1} + \frac{f^{(k)}(\xi)}{k!} \delta^{k}
\end{equation}
for a $\xi \in \left[x,x+\delta\right]$.\\
For $k=2$ specifically, this expression becomes:
\begin{equation}
\begin{cases}
f\left(x+\delta\right) = f(x) + f'(x) \delta + \frac{f''(\xi_+)}{2!} \delta^2,\\
f\left(x-\delta\right) = f(x) - f'(x) \delta + \frac{f''(\xi_-)}{2!} \delta^2,
\end{cases}
\end{equation}
for some $\xi_+,\xi_- \in \left[x,x+\delta\right]$.

The central difference method for numerical differentiation uses the following formula, derived from the expressions above:
\begin{equation}
f'(x) = \frac{f\left(x+\delta\right) - f\left(x-\delta\right)}{2\delta} + \frac{f''(\xi_+)-f''(\xi_-)}{4}\delta.
\end{equation}
When a bound $C_2$ for $f''$ is known on the interval $[x-\delta,x+\delta]$, the remainder term can be bounded by
\begin{equation}
\abs{\frac{f''(\xi_+)-f''(\xi_-)}{4}\delta} \leq \frac{C_2}{2} \delta.
\end{equation}
The method can be generalized to use higher order derivatives (up to some $k\in\N$), such that $f'(x)$ is now of the form
\begin{equation}
f'(x) = \sum_{l=-m}^{m} \underbrace{\frac{a_l^{(2m)} f(x+l\delta)}{\delta}}_{f_l} + \underbrace{\sum_{l=-m}^{m} a_l^{(2m)} \frac{f^{(k)}(\xi_l)}{k!} l^k\delta^{k-1}}_{R_k}
\end{equation}
for $m = \lfloor \frac{k-1}{2}\rfloor$ and where
\begin{equation}\label{eq:al}
a_l^{(2m)}=
\begin{cases}
1 & \text{if }l=0,\\
\frac{(-1)^{l+1}(m!)^2}{l(m+l)!(m-l)!} & \text{otherwise}.
\end{cases}
\end{equation}

\subsection{Bounding the errors}

When a bound $C_k$ for $f^{(k)}$ is known on the interval $[x-m\delta,x+m\delta]$, the remainder term $R_k$ can be bounded by
\begin{equation}\label{eq:remainder-order-k}
\abs{R_k} \leq \abs{\sum_{l=-m}^{m} a_l^{(2m)}l^k} \frac{C_k}{k!}\delta^{k-1} \leq 2m^k\frac{C_k}{k!}\delta^{k-1}
\end{equation}
where the last inequality comes from Theorem 3.4 in \cite{cornelissen19}.

In order for $\abs{R_k} \leq \frac{\varepsilon}{2}$, we then need
\begin{equation}
\delta \leq \sqrt[k-1]{\frac{k!\varepsilon}{4m^kC_k}}.
\end{equation}

We take
\begin{align}
	\delta &= \frac{2}{e}\left(\frac{\varepsilon}{4C_k}\right)^{1/k}\label{eq:delta}\\
	&\leq (2\pi k)^{\frac{1}{2k}}\frac{k}{me}\left(\frac{\varepsilon}{4C_k}\right)^{1/k}\\
	&\leq \left(\frac{\sqrt{2\pi k}(k/e)^k\varepsilon}{4m^k C_k}\right)^{1/k}\\
	&\leq \sqrt[k]{\frac{k!\varepsilon}{4m^kC_k}}\\
	&\leq \sqrt[k-1]{\frac{k!\varepsilon}{4m^kC_k}}.
\end{align}

Moreover, we are interested in the case where $f$ cannot be evaluated exactly, but rather when we have access to random samples whose expectation value is $f(x)$ (and are bounded by $C_0$). If we want to estimate each $f_l$, $l=-m, \ldots, m$, to precision $\frac{\varepsilon}{2k}$ (such that we get their sum to precision $\frac{\varepsilon}{2}$), it is sufficient to estimate each $f(x+l\delta)$ to precision $\frac{\varepsilon\delta}{a_l^{(2m)}2k}$. From Lemma \ref{lem:classical-MVMC}, we have that this requires a total number of queries (or samples) that scales as
\begingroup
\allowdisplaybreaks
\begin{align}
	\widetilde{\O}\left(\sum_{l=-m}^{m} \left(\frac{C_0 k a_l^{(2m)}}{\varepsilon\delta}\right)^2\right) &\leq \widetilde{\O}\left(\left(\frac{C_0 k}{\varepsilon\delta}\right)^2 \sum_{l=-m}^{m} \abs{a_l^{(2m)}}\right)\\
	&\leq \widetilde{\O}\left(\left(\frac{C_0 k}{\varepsilon\delta}\right)^2 \left(1+2\sum_{l=1}^{m}\frac{1}{l}\right)\right)\\
	&\leq \widetilde{\O}\left(\left(\frac{C_0 k}{\varepsilon\delta}\right)^2 \left(3+2\log(m)\right)\right)\\
	&\leq \widetilde{\O}\left(\left(\frac{C_0 k}{\varepsilon\delta}\right)^2\right)\label{eq:nb-samples}
\end{align}
\endgroup
where the first two inequalities follow from $a_0^{(2m)} = 1$ (Eq.~(\ref{eq:al})) and $\abs{a_l^{(2m)}} \leq \frac{1}{\abs{l}}, l\in\{-m,\ldots,m\}\backslash \{0\}$ (Theorem 3.4 in \cite{cornelissen19}), and the third inequality follows from a simple upper bound on harmonic numbers.

Combining Eqs.~(\ref{eq:delta}) and (\ref{eq:nb-samples}), we find that a total of
\begin{equation}\label{eq:query-complexity}
	\widetilde{\O}\left(\left(\frac{C_0k}{\varepsilon}\left(\frac{C_k}{\varepsilon}\right)^{1/k}\right)^2\right)
\end{equation}
queries are sufficient to estimate $f'(x)$ to precision $\varepsilon$. 

\subsection{Application to value functions}

In the case of value functions, we have $C_k = \frac{2\Rmax}{1-\gamma} \left(DT^2\right)^{k} \forall k\in\N$ (see Lemma \ref{lem:gevrey-valuefct}). Therefore, we can choose
\begin{equation}
	k = \log \left(\frac{2\Rmax}{\varepsilon(1-\gamma)}\right)
\end{equation}
and use the identity $x^{1/\log(x)} = e^{\log(x)/\log(x)}=e$, such that, from Eq.~(\ref{eq:query-complexity}):
\begin{equation}
	\widetilde{\O}\left(d\left(\frac{DT^2\Rmax}{\varepsilon(1-\gamma)}\right)^2\right)
\end{equation}
queries are sufficient to estimate the gradient $\grad\valuefct$ to $\varepsilon$ precision in the $\ell_\infty$-norm. The multiplicative factor $d$ comes from the fact that we need to estimate each of the $d$ coordinates of the gradient independently.

\end{document}